\documentclass{article}
\usepackage{amsmath, amssymb, amsthm}
\usepackage{xcolor}
\usepackage{graphicx}
\usepackage[vlined,ruled,linesnumbered]{algorithm2e}
\usepackage{booktabs}
\usepackage{tikz}
\usetikzlibrary{automata,positioning}

\usepackage{todo}
\usepackage{cancel}

\newtheorem{definition}{Definition}
\newtheorem{theorem}{Theorem}

\newcommand{\Dslash}{\mathcal{D}}
\newcommand{\mathC}{\mathbb{C}}
\newcommand{\calL}{\mathcal{L}}

\newcommand{\lc}{\mathrm{lc}}
\newcommand{\SAP}{\mathrm{SAP}}

\newcommand{\SetAlgorithmStyle}{
  \SetKwInput{Input}{input}
  \SetKwInput{Output}{output}
  \SetKwComment{tcpS}{\{}{\phantom{\}}}
  \SetKwComment{tcpM}{}{\phantom{\}}}
  \SetKwComment{tcpE}{}{\}}
  \SetKwFor{ForAll}{for all}{do}{end}
  \SetKwFor{For}{for}{}{end}
  \SetKwFor{While}{while}{do}{end}
  \SetKwFor{If}{if}{}{end}
  \SetArgSty{}
  \DontPrintSemicolon
}

\numberwithin{equation}{section}

\newtheorem{proposition}{Proposition}
\title{A multigrid accelerated eigensolver for the Hermitian Wilson-Dirac operator in lattice QCD}
\author{Andreas Frommer, Karsten Kahl, Francesco Knechtli, \\ Matthias Rottmann, Artur Strebel, Ian Zwaan}

\date{}

\begin{document}
\maketitle

\begin{abstract}
Eigenvalues of the Hermitian Wilson-Dirac operator are of special interest in several lattice QCD simulations, e.g., for noise reduction when evaluating all-to-all propagators.
In this paper we present a Davidson-type eigensolver that utilizes the structural properties of the Hermitian Wilson-Dirac operator $Q$ to compute eigenpairs of this operator corresponding to small eigenvalues.
The main idea is to exploit a synergy between the (outer) eigensolver and its (inner) iterative scheme which solves shifted linear systems. This is achieved by adapting the multigrid DD-$\alpha$AMG algorithm to a solver for shifted systems involving the Hermitian Wilson-Dirac operator.
We demonstrate that updating the coarse grid operator using eigenvector information obtained in the course of the generalized Davidson method is crucial to achieve good performance when calculating many eigenpairs, as our study of the local coherence shows.
We compare our method with the commonly used software-packages PARPACK and PRIMME in numerical tests, where we are able to achieve significant improvements, with speed-ups of up to one order of magnitude and a near-linear scaling with respect to the number of eigenvalues. For illustration we compare the distribution of the small eigenvalues of $Q$ on a $64\times 32^3$ lattice with what is predicted by the Banks-Casher relation in the infinite volume limit.   
\end{abstract}

\section{Introduction}\label{sec:intro}
In lattice Quantum Chromodynamics (QCD) the Wilson-Dirac operator describes the interaction between quarks and gluons in the framework of quantum field theory. Results of lattice QCD simulations represent essential input to several of the current and planned experiments in elementary particle physics (e.g., BELLE II, LHCb, EIC, PANDA, BES III).

Obtaining eigenpairs (eigenvalues and -vectors) of the Wilson-Dirac operator is an important computational task. For example, eigenpairs can be used to directly compute physical observables~\cite{Blossier:2010vz, DeGrand:2004qw, Endress:2014qpa, Foley:2005ac, Giusti:2004yp,  Neff:2001zr} or as a tool for noise reduction in stochastically estimated quantities like disconnected fermion loops~\cite{BaliSimeth}.
In most circumstances we are interested in a small to moderate amount of eigenvectors corresponding to the eigenvalues closest to zero, especially for the Hermitian Wilson-Dirac operator. As the Hermitian Wilson-Dirac operator is indefinite, these eigenvalues lie in the \emph{interior} of the spectrum.

Typically, computing interior eigenvalues is particularly expensive, which is why in this paper we develop efficient computational methods for the special case of the Hermitian Wilson-Dirac operator.
The most prominent methods for obtaining interior eigenvalues are \emph{shift-and-invert} algorithms which extend the basic inverse iteration approach; cf.~\cite{Saad11}. This includes methods ranging from the classical Rayleigh quotient iteration (RQI)~\cite{Saad11} to the generalized Davidson (GD) methods~\cite{Saad11} and its numerous variations like $\mbox{GD}+k$~\cite{Stathopoulos1}, Jacobi-Davidson (JD)~\cite{jacobidavidson} or JDCG/JDQMR~\cite{notay02, Stathopoulos1}. 
The generalized Davidson methods can, alternatively, also be regarded as a generalization of Arnoldi's method~\cite{morganscott}\footnote{In~\cite{morganscott} the authors relate generalized Davidson with the Lanczos method, but the statement also holds for non-Hermitian matrices.} with improved search directions.

Most Davidson-type methods share an \emph{inner-outer-scheme}, where the outer iteration finds approximations to the sought eigenpairs, while the inner iterations generates new search directions by approximately solving shifted linear systems
\begin{equation}\label{eq:shiftsolve}
  (A-\tau I)x = b,
\end{equation}
where $\tau$ is an approximation to a target eigenvalue.
In fact, these inversions make up the bulk of the computational work, and it is thus mandatory to find particularly efficient methods for this task.
In lattice QCD, adaptive algebraic multigrid methods have established themselves as the most efficient methods for solving linear systems with the Wilson-Dirac operator~\cite{MGClark2010_1,MGClark2007,FrKaKrLeRo11,Frommer:2013fsa,MGClark2010_2}.
They demonstrate significant speed-ups compared to conventional Krylov subspace methods, achieving orders of magnitudes faster convergence and an insensitivity to conditioning.
In this work, we use the adaptive domain decomposition algebraic multigrid method DD-$\alpha$AMG~\cite{FrKaKrLeRo11, Frommer:2013fsa}, but we expect the results for DD-$\alpha$AMG to carry over to the other aggregation-based multigrid implementations as well. {Originally DD-$\alpha$AMG is composed of an adaptive aggregation based multigrid construction and a red-black multiplicative Schwarz smoother (traditionally termed ``SAP'' for Schwarz alternating procedure in lattice QCD), but we also have the option to use other smoothers like GMRES in the DD-$\alpha$AMG framework.
}

So far, multigrid solvers have been limited to the unshifted Wilson-Dirac operator\footnote{The original solver can manage small shifts, but quickly deteriorates for increasing shifts as will be demonstrated later.}. This is due to the algebraic construction of the interpolation operator, which is built from approximations of eigenvectors corresponding to small eigenvalues, which is necessary for an effective 
overall error reduction in the unshifted case, cf.~\cite{FrKaKrLeRo11, Frommer:2013fsa}. 

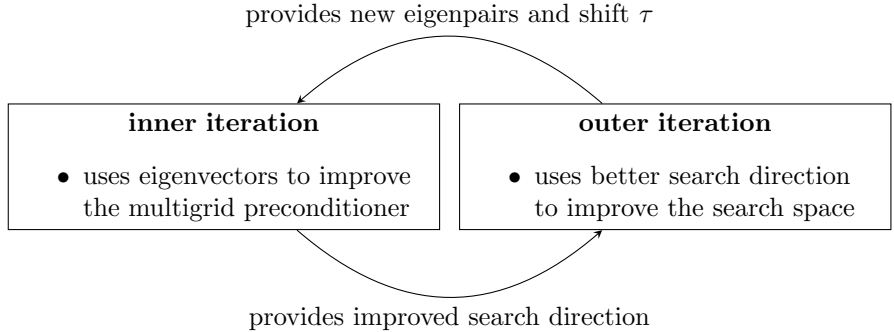
\begin{figure}[htb]
\centering
\begin{tikzpicture}[%
    >=stealth,
    node distance=6cm,
    on grid,
    auto
  ]
    \node[rectangle,draw] (A) [text width=5.5cm] {\centering\textbf{inner iteration}\begin{itemize} \item uses eigenvectors to improve the multigrid preconditioner\end {itemize}}; 
    \node[rectangle,draw] (B) [right of=A,text width=5.5cm]{\centering\textbf{outer iteration}\begin{itemize} \item uses better search direction to improve the search space\end{itemize}}; 
    \path[->, yscale=1.5] (A) edge[bend right] node [below] {provides improved search direction} (B);
    \path[->, yscale=1.5] (B) edge[bend right] node [above] {provides new eigenpairs and shift $\tau$} (A);
\end{tikzpicture}
\label{fig:methodsketch}
\caption{Interleaving eigenpair extraction and construction of the multigrid preconditioner for the inner iteration}
\end{figure}

During the progress of the eigensolver---when the shift $\tau$ in \eqref{eq:shiftsolve} becomes larger---the interpolation operator no longer approximates the space spanned by eigenvectors corresponding to small eigenvalues of the shifted operator, thus invalidating the coarse grid correction step and significantly slowing down convergence.
The main idea of this paper is that in order to overcome this problem we dynamically update the interpolation operator of the multigrid solver during the outer iteration.
This idea is sketched in Figure~\ref{fig:methodsketch}.
It illustrates how we are interleaving the eigenpair extraction with the construction of an efficient solver for \eqref{eq:shiftsolve}. While the outer iteration extracts eigenpairs, its eigenvectors are used to improve the multigrid preconditioner. In turn the multigrid method is an efficient preconditioner for \eqref{eq:shiftsolve}, which produces improved search directions for the outer iteration.

As an additional topic, we investigate several approaches for the most suitable smoothing method of the multigrid method for the Hermitian Wilson-Dirac operator in presence of large shifts.
{Altogether, we obtain a method which scales close to linearly with the number of eigenpairs to be computed.}

The remainder of this paper is organized as follows. 
We give a brief introduction into lattice QCD and algebraic multigrid methods in Section~\ref{sec:ddamg}. We put particular emphasis on showing and theoretically justifying how the algebraic multigrid approach, which has originally been developed for the non-Hermitian Wilson-Dirac operator, can be applied to the Hermitian Wilson-Dirac operator.
In Section~\ref{sec:eig} we proceed by introducing the generalized Davidson method and present our adaptations for the specific case of the Hermitian Wilson-Dirac operator, including the interleaving of the eigensolver and the construction of the multigrid solver used in the inner iteration. 
Numerical tests and comparisons with commonly used and state-of-the-art software are presented in Section~\ref{sec:tests}. 
As a simple example of the use of our method in physics, Section~\ref{sec:gap} discusses lattice artifacts observed for the spectral gap. Finally, a summary of our results is given in Section~\ref{sec:conclusion}. 

\section{Algebraic Multigrid Methods in Lattice QCD}\label{sec:ddamg}

The Dirac equation
\begin{equation}\label{Dirac_eq}
  \Dslash\psi + m \cdot \psi = \eta
\end{equation} 
describes the dynamics of quarks and the interaction of quarks and gluons.
Here, $\psi = \psi(x)$ and $\eta = \eta(x)$ represent quark fields. 
They depend on $x$, the points in space-time, $x=(x_0,x_1,x_2,x_3)$. 
The gluons are represented in the Dirac operator $\Dslash$, and $m$ is a scalar mass parameter that is independent of $x$ and sets the mass of the quarks in the QCD theory.

More precisely, $\Dslash$ is given as 
\begin{equation} \label{Dirac_continuum:eq}
\Dslash=\sum_{\mu=0}^3\gamma_\mu \otimes \left( \partial_\mu + A_\mu \right)\,,
\end{equation}
where $ \partial_\mu = \partial / \partial x_\mu$ and $A$ is the gluon (background) gauge field with 
the anti-Hermitian traceless matrices $A_\mu(x)$ being elements of $\mathfrak{su}(3)$, the Lie 
algebra of the special unitary group $\mathrm{SU}(3)$. The (Hermitian) $\gamma$-matrices 
$\gamma_0,\gamma_1,\gamma_2,\gamma_3 \in  \mathbb{C}^{4 \times 4}$ represent  generators of the 
Clifford algebra with
\begin{equation} \label{commutativity_rel:eq}
  \gamma_\mu \gamma_\nu + \gamma_\nu \gamma_\mu = 2\delta_{\mu\nu}I_4 \quad \text{ for } \mu,\nu=0,1,2,3,
\end{equation}
with $I_4$ the identity on $\mathC^4$.
Consequently, at each point $x$ in space-time, the spinor $\psi(x)$, i.e., the quark field $\psi$ at 
a given point $x$, is a twelve component column vector, each component corresponding to one of 
three colors (acted upon by $A_\mu(x)$) and four spins (acted upon by $\gamma_\mu$).  
For future use we remark that $\gamma_5 = 
\gamma_0\gamma_1\gamma_2\gamma_3$ satisfies
\begin{equation} \label{gamma_commutativity:eq}
     \gamma_5 \gamma_\mu = - \gamma_\mu \gamma_5, \enspace \mu=0,1,2,3,
\end{equation}
independently from the chosen representation.

The only known way to obtain predictions in QCD from first principles and non-perturbatively, is to discretize and then simulate on a computer. The discretization is typically formulated on an equispaced periodic $N_t \times N_s^3$ lattice $\mathcal{L}$ with uniform lattice spacing $a$, $N_s$ denoting the number of lattice points for each of the three spatial dimensions and $N_t$ the number of lattice points in the time dimension. A quark field $\psi$ is now represented by its values at each lattice point, i.e., it is
a spinor valued function $\psi: x \in \mathcal{L} \to \psi(x) \in \mathbb{C}^{12}$.

The {\em Wilson-Dirac} discretization is one of the most commonly used discretizations in lattice QCD simulations. It is obtained from the continuum equation by replacing the covariant derivatives by centralized covariant finite differences on the lattice. 
It contains an additional, second order finite difference stabilization term, as 
otherwise the discretization would suffer from `red-black' instability, cf.~\cite{smith1985numerical}.
The Wilson-Dirac discretization yields a local operator $D$ in the sense 
that it represents a nearest neighbor coupling on the lattice $\mathcal{L}$. 

Introducing shift vectors $ \hat{\mu} = 
(\hat{\mu}_0,\hat{\mu}_1, \hat{\mu}_2, \hat{\mu}_3)^T  \in \mathbb{R}^4 $ in 
dimension $\mu$ on $\mathcal{L}$, i.e., 
$$
  \hat{\mu}_{\nu} = \begin{cases} a & \mu=\nu \\ 0 & \text{else} \end{cases},
$$
the action of $D$ on a discrete quark field $\psi$ is given as
\begin{eqnarray}
  \hspace*{-0.75em}(D\psi)(x) = (m_0+\frac{4}{a}) \psi(x) 
              &\hspace*{-0.75em}-&\hspace*{-0.75em} \frac{1}{2a}\sum_{\mu=0}^3 \left( (I_4-\gamma_\mu)\otimes U_\mu(x)\right) \psi(x+\hat{\mu}) \nonumber \\
             &\hspace*{-0.75em}-&\hspace*{-0.75em} \frac{1}{2a}\sum_{\mu=0}^3 \left( (I_4+\gamma_\mu)\otimes U_\mu^H(x-\hat{\mu})\right) \psi(x-\hat{\mu}). \label{Wilson-Dirac:eq}
\end{eqnarray}
Here, the gauge-links $U_\mu(x)$ are now matrices from the Lie group SU(3), and the lattice indices $x\pm \hat{\mu}$ are to be understood periodically.
The mass parameter $m_0$ sets the quark mass (for further details, see~\cite{montvay1994quantum}), and we will write $D(m_0)$ whenever the dependence on $m_0$ is important. 

To explicitly describe $D$ we fix a representation for the $\gamma$-matrices in which $\gamma_5 = \left( \begin{smallmatrix} I_2 & 0 \\ 0 & -I_2 \end{smallmatrix} \right)$.

From \eqref{Wilson-Dirac:eq} we obtain the
couplings of the lattice sites $x$ and $x\pm\hat{\mu}$ as
\begin{equation} \label{DandDH_entries:eq}
(D)_{x,x+\hat{\mu}} = \tfrac{-1}{2a}(I_4-\gamma_\mu) \otimes U_\mu(x), \enspace 
(D)_{x,x-\hat{\mu}} = \tfrac{-1}{2a}(I_4+\gamma_\mu) \otimes U^H_\mu(x-\hat{\mu}),
\end{equation}
which shows
$
(D)_{x+\hat{\mu},x} = \tfrac{-1}{2a}(I_4+\gamma_\mu) \otimes U^H_\mu(x).
$
Thus, the commutativity relations \eqref{gamma_commutativity:eq} imply the symmetry 
\[
(\gamma_5\otimes I_3) \big(D\big)_{x,x+\hat{\mu}} = \big((\gamma_5\otimes I_3) \big(D\big)_{x+\hat{\mu},x}\big)^H. 
\]
With $\Gamma_5 = I_{n_{\mathcal{L}}} \otimes \gamma_5 \otimes I_3 $, $n_\mathcal{L}$ the number of lattice sites, this symmetry can be described on the level of the entire Wilson-Dirac operator as 
\begin{equation} \label{gamma_5_symmetry:eq}
    \Gamma_5 D = (\Gamma_5 D)^H.
\end{equation}
The matrix $\Gamma_5$ is Hermitian and unitary, and the $\Gamma_5$-symmetry \eqref{gamma_5_symmetry:eq} is a non-trivial, fundamental symmetry that the discrete Wilson-Dirac operator inherits from a corresponding symmetry of the continuum Dirac operator \eqref{Dirac_continuum:eq}.

The Wilson-Dirac operator and its clover-improved variant (where a term which is diagonal in space and time is added to reduce the local discretization error from $\mathcal{O}(a)$ to $\mathcal{O}(a^2)$) is an adequate discretization for the numerical computation of many physical observables.
For further details on discretization, its properties and the clover-improved variant, we refer the interested reader to~\cite{DeGrand:2006zz,Gattringer:2010zz,Sheikholeslami:1985ij}. In this paper we focus on the \emph{Hermitian} or \emph{symmetrized Wilson-Dirac operator} 
$$ Q:=\Gamma_5 D \, .$$

\paragraph{Algebraic multigrid methods.}
The state-of-the-art approaches for solving linear systems involving the (non-Hermitian) Wilson-Dirac operator $D$ are variants of aggregation-based adaptive algebraic multigrid methods, see~\cite{MGClark2010_1,MGClark2007,Frommer:2013fsa,MGClark2010_2}. 
Typically, these multigrid solvers are used as a (non-stationary) preconditioner within a flexible Krylov subspace method like FGMRES~\cite{fgmres} or GCR \cite{gcr}; see also~\cite{Luscher:2003qa,Luescher2007}.

The error propagator for the two-level version of all these multigrid approaches is 
\begin{equation}
  E_{2g} = (I-MD)^{\nu}(I-P D_{c}^{-1} R D)(I-MD)^{\mu},
\end{equation}
where $M$ denotes the smoother---which, in the case of DD-$\alpha$AMG, is given by the Schwarz alternating procedure (SAP)---and $\mu$ and $\nu$ denote the number of pre- and post-smoothing iterations, respectively. 
The operator $I-PD_c^{-1}RD$ is the coarse grid correction, where $P$ is the adaptively constructed aggregation based interpolation~\cite{MGClark2010_1,MGClark2007,Frommer:2013fsa,MGClark2010_2}, obtained in a ``setup'' phase, $R = P^H$ is the corresponding restriction and $D_c$ the Galerkin projected coarse grid operator $D_c = P^HDP$.

As is discussed in~\cite{BranKahl2014,MGClark2010_2}, this algebraic multigrid approach for $D$ can be transferred to one for $Q$ if the interpolation $P$ preserves spin structure in the sense that on the coarse grid we can partition the degrees of freedom per grid point into two groups corresponding to different spins and that we have  $\Gamma_5 P = P \Gamma_5^c$, where $\Gamma_5^c$ is diagonal with values $\pm 1$, depending on the spin on the coarse grid. Putting $Q_c = \Gamma_5^cD_c$ we then have 
\begin{equation}\label{eq:coarsegrid_for_Q}
  I-P Q_{c}^{-1} P^H Q = I-P D_c^{-1} \Gamma_5^c P^H \Gamma_5 D = I-P D_c^{-1} P^H D \, ,
\end{equation}
showing that the coarse grid error propagator for $D$ is identical to the coarse grid error propagator for $Q$ if we take the same $P$.  Note that the construction of $P$ in~\cite{MGClark2010_1,Frommer:2013fsa,MGClark2010_2} is indeed spin structure preserving, while this is not the case for the ``little Dirac'' construction found in \cite{Luescher2007}.

As a matter of fact the SAP smoothing used in DD-$\alpha$AMG is identical for $D$ and $Q$, as well, which can be shown by the following argument.
Mathematically, one step of SAP is a product of block projections, i.e., the error propagator is given by 
\begin{equation}
  E_{\SAP}:=\prod_{i=1}^b(I - \underbrace{I_{\calL_i} Q_i^{-1} I_{\calL_i}^H}_{:=M_{Q_i}}Q), 
\end{equation}
where $b$ is the number of subdomains, $\calL_i$ is the $i$-th subdomain of the lattice $\calL$, $I_{\calL_i}$ the trivial injection from $\calL_i$ into $\calL$, and $Q_i:= I_{\calL_i}^H Q I_{\calL_i}$ the block restriction of $Q$ on $\calL_i$.

Note that algorithmically, the calculations corresponding to $I_{\calL_i} Q_i^{-1} I_{\calL_i}^HQ$ can be performed in parallel for all blocks $i$ of the same color if we introduce a red-back ordering on the blocks.

With this we get the following proposition, in which we define $M_{D_i}$ and $D_i$ analogously to $M_{Q_i}$. 
\begin{proposition}~\label{thm:sap_smoothing}
The error propagator $E_{\SAP}(Q):=\prod_{i=1}^b(I - M_{Q_i}Q)$ is equivalent to $E_{\SAP}(D):=\prod_{i=1}^b(I - M_{D_i}D)$.
\end{proposition}
\begin{proof} We first note that $\Gamma_5$ is just a local positive or negative identity, so its block restriction $\Gamma_5^i := I_{\calL_i}^H\Gamma_5 I_{\calL_i}$ on $\calL_i$ not only satisfies $I_{\calL i}^H\Gamma_5 = \Gamma_5^i I_{\calL_i}^H$ but also $(\Gamma_5^i)^{-1} = \Gamma_5^i$. To prove the proposition we only need to show that the error propagators are identical for any given subdomain $i$:
\begin{equation}
\begin{aligned}
  I-M_{Q_i} Q &= I-(I_{\calL_i} Q_i^{-1} I_{\calL_i}^H) Q \\
          &= I-(I_{\calL_i} (I_{\calL_i}^H \Gamma_5 D I_{\calL_i})^{-1} I_{\calL_i}^H) \Gamma_5 D \\
          &= I-(I_{\calL_i} (\Gamma_5^i D_i)^{-1} I_{\calL_i}^H) \Gamma_5 D \\
          &= I-(I_{\calL_i} D_i^{-1} \Gamma_5^i I_{\calL_i}^H) \Gamma_5 D \\
          &= I-(I_{\calL_i} D_i^{-1} I_{\calL_i}^H \Gamma_5)\Gamma_5 D = I-M_{D_i} D.
\end{aligned}\label{eq:sap_smoothing}
\end{equation}
\end{proof}

The proposition states that SAP for $Q$ is equivalent to SAP for $D$ if the block inversions for the block systems $Q_i$ are performed exactly, which together with~\eqref{eq:coarsegrid_for_Q} implies that the DD-$\alpha$AMG method has the same error propagator, irrespective of whether it is applied to $Q$ or to $D$.
As observed in~\cite{Frommer:2013fsa} SAP smoothing works well for the standard Wilson-Dirac operator $D$ thus it also works well for the Hermitian Wilson-Dirac operator $Q$. However, if we perform only approximate block inversions in SAP---and this is what one typically does--- this situation becomes less clear; see Section~\ref{sec:tests}. 

Alternatively, instead of SAP one can use (restarted) GMRES as a  smoother for $Q$. For the non-Hermitian Wilson-Dirac operator $D$ this is used in the multigrid methods from~\cite{MGClark2010_1,MGClark2007,MGClark2010_2}, and since GMRES is also one of the most numerically stable Krylov subspace methods for indefinite systems, it is to be expected to work well as a smoother in a multigrid method for $Q$ as well.
Interestingly, for GMRES smoothing a connection between $Q$ and $D$ similar to what has just been exposed for SAP smoothing does not hold. We compare the above options for the smoothing method experimentally in Section~\ref{sec:tests}.

\section{Eigensolver}\label{sec:eig}

The \emph{generalized Davidson} (GD) method~\cite{Davidson,morganscott,Saad11} is an eigensolver framework which can be seen as a generalization of Arnoldi's method. Its advantage is that it does not rely on a Krylov subspace structure and thus offers a more flexible way of steering the search space $\mathcal{V}_m$ into a desired direction.
The method successively generates a set of orthogonal vectors $v_1,v_2,\ldots,v_m,$ which span the search space $\mathcal{V}_m$. An approximate eigenpair $(u,\theta)$ with $u \in \mathcal{V}_m$ is chosen such that the Ritz-Galerkin condition
\begin{equation}
  Au-\theta u\perp \mathcal{V}_m
\end{equation}
holds, which amounts to solving the (small and dense) $m\times m$ eigenvalue problem
\begin{equation} \label{eq:ritz}
  \left(V_m^HAV_m\right)s-\theta s = 0,  \text{ where } V_m = [v_1 \mid \cdots \mid v_m],  
\end{equation}
and then taking $u = V_m s$ with $s$ the eigenvector from \eqref{eq:ritz} whose eigenvalue $\theta$ is closest to the target eigenvalue. The search space is then extended by a new vector $t$ which is obtained as a function of the matrix $A$, {the approximate eigenvalue $\theta$}, and 
the eigenvector residual $r:=Au-\theta u$. The new vector $v_{m+1}$ is then retrieved after orthogonalizing $t$ against $v_1,\ldots,v_m$ and normalizing it.

For our work we focus on obtaining $t$ as an (approximate) solution of the \emph{correction equation} 
\begin{equation} \label{eq:correction}
  (A-\tau I)t = r,  
\end{equation}
where $\tau$ is an estimate for the target eigenvalue. 
The choice of $\tau$ steers the expansion of the search space, and with it the Ritz values, towards the desired eigenvalue regions, e.g., eigenvalues with smallest absolute value or with largest imaginary part.

For the (Hermitian) Wilson-Dirac operator, Davidson-type methods are to be preferred over Arnoldi's method due to the fact that Arnoldi's method would require exact solves of the correction equation to maintain its constitutive orthogonality relations, whereas Davidson-type
methods are tailored to accommodate approximate solutions, and these can be computed efficiently via some steps of multigrid preconditioned flexible GMRES.

A description of the generalized Davidson method to obtain one eigenpair is given in Algorithm~\ref{alg:gd1}.
Techniques for computing several eigenpairs and restarting will be reviewed in the subsequent section.

\begin{algorithm}[ht]
  \caption{Generalized Davidson (basic)}\label{alg:gd1}
  \SetAlgorithmStyle
    \Input{initial guess $t$, desired accuracy $\varepsilon$}
    \Output{eigenpair $(\lambda,x)$}
    $V=\emptyset$\;
    \For{$m=1,2,\ldots$} {
      $t = (I - VV^H)t$\; \label{alg:gd1:orth1}
      $v_m = t/||t||_2$\; 
      $V = [V \mid v_m]$ \; \label{alg:gd1:orth2}
      $H = V^HAV$\;
      get target eigenpair $(\theta,s)$ of $H$\;
      $u=Vs$\;
      $r = Au-\theta u$\;
      \If{$||r||_2 \leq\varepsilon$} {
        $\lambda = \theta$, $x=u$\; \Return\;
      }
      {compute $t$ as a function of $A$, $r$ and $\theta$} \label{alg:gd1:prec} 
    }
\end{algorithm}

\subsection{GD-$\lambda$AMG}\label{sec:gdlamg}
The method we propose for the Hermitian Wilson-Dirac operator is based on Algorithm~\ref{alg:gd1} but incorporates several adaptations for the Hermitian Wilson-Dirac operator and the underlying DD-$\alpha$AMG multigrid solver.

A first challenge is that we are confronted with a ``maximally indefinite'' interior eigenvalue problem, seeking the eigenvalues closest to zero, while the operator has a nearly equal amount of positive and negative eigenvalues.
The basic generalized Davidson method uses the Rayleigh Ritz procedure to determine the Ritz approximation by solving the standard eigenvalue problem \eqref{eq:ritz} for $H  = V_m^H A V_m$.
Ritz values approximate outer eigenvalues better and faster than the interior ones \cite{Saad11}, which is why we use \emph{harmonic Ritz values}~\cite{paige95} instead. 

\begin{definition}[Harmonic Ritz Values]
  A value $\theta\in\mathC$ is called a \emph{harmonic Ritz value} of $A$ with respect to a linear subspace $\mathcal{V}$ if $\theta^{-1}$ is a Ritz value of $A^{-1}$ with respect to $\mathcal{V}$.
\end{definition}

As the exterior eigenvalues of $A^{-1}$ are the inverses of the eigenvalues of $A$ of small modulus, harmonic Ritz values tend to approximate small eigenvalues well. Inverting $A$ to obtain harmonic Ritz values can be avoided with an appropriate choice for $\mathcal{V}$ as stated in the following theorem; cf.~\cite{jacobidavidson}.

\begin{theorem}\label{thm:harmonic}
  Let $\mathcal{V}$ be some $m$-dimensional subspace with basis $v_1,\ldots, v_m$.
  A value $\theta\in\mathC$ is a harmonic Ritz value of $A$ with respect to the subspace $\mathcal{W}:=A\mathcal{V}$, if and only if
  \begin{equation} \label{harmonic_Ritz:eq}Au_m-\theta u_m\perp A\mathcal{V} \text{ for some } u_m\in\mathcal{V}, u_m\neq 0.
  \end{equation}
  With 
  $$
  V_m:=[v_1|\ldots|v_m], \; W_m:= AV_m
  \text{ and } H_m := (W_m^HV_m)^{-1}W_m^HAV_m,
  $$
  \eqref{harmonic_Ritz:eq} is equivalent to 
  $$ H_ms=\theta s\text{ for some }s\in\mathC^{m},s\neq0\text{ and } u_m=V_ms.$$
\end{theorem}

Due to Theorem~\ref{thm:harmonic}, we can obtain harmonic Ritz values by solving the generalized eigenvalue problem
\begin{equation}
  W^H_m A V_m u = \theta W_m^H V_m u.
\end{equation}

The computational overhead compared to the standard Ritz procedure is dominated by $m^2$ additional inner products to build $W_m^H A V_m$. In our numerical tests, we have observed that this is compensated by a faster convergence of the generalized Davidson method, cf. Section~\ref{sec:tests}.

Although the multigrid approach is viable for the Hermitian Wilson-Dirac operator $Q$, it is, in practice, slower than for $D$. For exact solves of the subdomain systems in the SAP smoother, the discussion in Section~\ref{sec:ddamg} and Proposition~\ref{thm:sap_smoothing} implies that the convergence speeds for $Q$ and for $D$ are comparable as the error propagation operators are identical.
Though in computational practice, it is more efficient to do only approximate solves for the subdomain systems, using a small number of GMRES steps, for example. In this scenario the multigrid method becomes significantly slower when used for $Q$ rather than $D$, see Figure~\ref{fig:smoother} in Chapter~\ref{sec:tests}.
This slowdown can be countered by left-preconditioning the correction equation with $\Gamma_5$. This means that instead of solving~\eqref{eq:shiftsolve} with $Q$, we can transform it equivalently according to 
\begin{eqnarray}
                      &(Q-\tau I)t         &= r         \label{eq:no_g5_prec}\\
  \Longleftrightarrow &\Gamma_5(Q-\tau I)t &= \Gamma_5r \nonumber\\
  \Longleftrightarrow &(D-\tau \Gamma_5)t  &= \Gamma_5r.\label{eq:g5_prec}
\end{eqnarray}
The spectrum of the resulting operator $\Gamma_5Q(\tau):=D-\tau \Gamma_5$ has similarities to that of $D$ with some eigenvalues collapsing on the real axis. As we will see in Chapter~\ref{sec:tests}, this simple transformation speeds up the multigrid method significantly.
Figure~\ref{fig:specs} shows full spectra of $D$, $Q$ and $\Gamma_5Q(\tau)$ for a configuration on a small $4^4$ lattice.

\begin{figure}[ht]
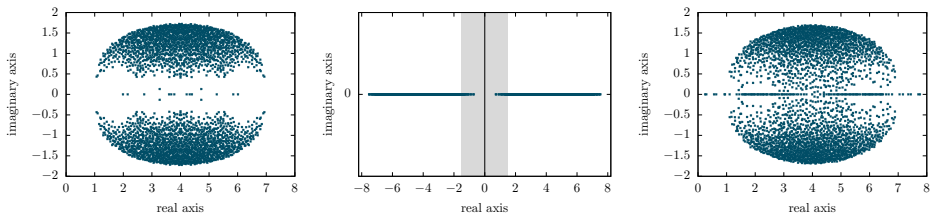

  \hspace{-1.8em}\scalebox{0.49}{\input{./plots/specD}}\scalebox{0.49}{\input{./plots/specQ}}\scalebox{0.49}{\input{./plots/specDplusG5shift}}
  \caption{Full Spectra of $D$, $Q$ and $\Gamma_5Q(\tau)$ for configuration $5$ (see Table~\ref{tab:confs}). The shaded area of $\operatorname{spec}(Q)$ highlights the eigenvalues we are particularly interested in.} 
  \label{fig:specs}
\end{figure}

\paragraph{Restarting and locking.}\label{sec:restart}

As the search space grows in every outer iteration, the storage and orthogonalization costs of the outer iteration in a generalized Davidson method eventually become prohibitively large. The following techniques reduce these costs in order to achieve a near-linear scaling in the number of computed eigenpairs.
The first technique is \emph{thick restarting}~\cite{jacobidavidson}.
When the search space reaches a size of $m_{\mathit{max}}$, we perform a restart by discarding the current search space. At the same time we keep the first $m_{\mathit{min}}$ smallest non-converged harmonic Ritz vectors and use them to span the search space at the beginning of the next restart cycle. We tuned the parameters $m_{\mathit{min}}$ and $m_{\mathit{max}}$ such that we have reason to assume that we keep both positive and negative harmonic Ritz values within the new search space. This way the eigensolver obtains a (nearly) equal amount of positive and negative eigenpairs in a uniform way.

In order to avoid re-targeting converged eigenpairs, we employ the concept of \emph{locking} converged eigenpairs~\cite{Stathopoulos2} as a second technique. Locking keeps the search space $\mathcal{V}$ orthogonal to the space of already converged eigenvectors $\mathcal{X}$.
In this manner, it is not required to keep converged eigenvectors in the search space which has the effect that the search space dimension becomes bounded independently of the number of eigenpairs sought. This in turn bounds the cost for computing the harmonic Ritz pairs. The new search direction still has to be orthogonalized against all previous eigenvectors, which leads to costs of order $\mathcal{O}(nk^2)$, since it consists of $k-1$ inner products for each of the $k$ eigenpairs. This is responsible for the fact that, in principle, the cost of our method scales superlinearly with $k$, and this becomes visible when $k$ becomes sufficiently large.

\paragraph{Local coherence and its effect on the correction equation.} \label{sec:update_interpolation}

The strength of algebraic multigrid methods relies on an effective coarse grid correction step and thus on the construction of the interpolation operator $P$. The methods in use for the Wilson-Dirac operator are all adaptive:  They require a setup phase which computes ``test vectors'' $w_i, i=1,\ldots,n_{\text{\emph{tv}}}$ which are collected as columns in the matrix $W = [w_1\mid\ldots\mid w_{n_{\text{\emph{tv}}}}]$. The test vectors are approximations to eigenvectors corresponding to small eigenvalues of the unshifted Wilson-Dirac operator $D$. The matrix $W$ is then used to build an aggregation based, ``block diagonal'' interpolation operator $P$, where each diagonal block constitutes an \emph{aggregate}, i.e., a block $\mathcal{A}_i$ of $W$ corresponding to the degrees of freedom of a block of the lattice $\calL$; see Figure~\ref{fig:buildP} and \cite{Frommer:2013fsa,MGClark2010_2}.

\begin{figure}[ht]
  \center\scalebox{0.5}{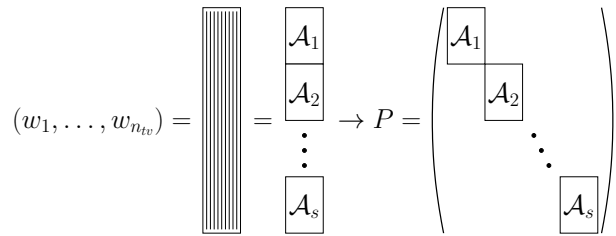}
\caption{Matrix view of the construction of the aggregation based interpolation operator $P$.} \label{fig:buildP}
\end{figure}

By construction, the range of an aggregation based interpolation $P$ contains \emph{at least} the range spanned by the test vectors it is being built from. In~\cite{Luescher2007} it has been observed that eigenvectors belonging to small eigenvalues of the Wilson-Dirac operator $D$ are \emph{locally coherent} in the sense that these eigenvectors are locally similar, i.e., they are similar on the individual aggregates. This is the reason why the span of an aggregation based interpolation $P$ contains good approximations to small eigenpairs {\em far beyond} those which are explicitly used for its construction. This in turn explains the efficiency of such $P$ in the multigrid method.

We can study local coherence using the \emph{local coherence measure} $\lc$ of a vector $v$ defined as
\[
\lc(v) = \| \Pi v\|/\|v\|,
\]
where $\Pi$ denotes the orthogonal projection on the range of $P$. 
If $\lc(v)$ is close to 1, there is a good approximation to $v$ in the range of $P$, implying that the multigrid coarse grid correction reduces error components in the direction of $v$ almost to zero.

\begin{figure}[ht]
  \begin{center}
    \scalebox{0.25}{\includegraphics{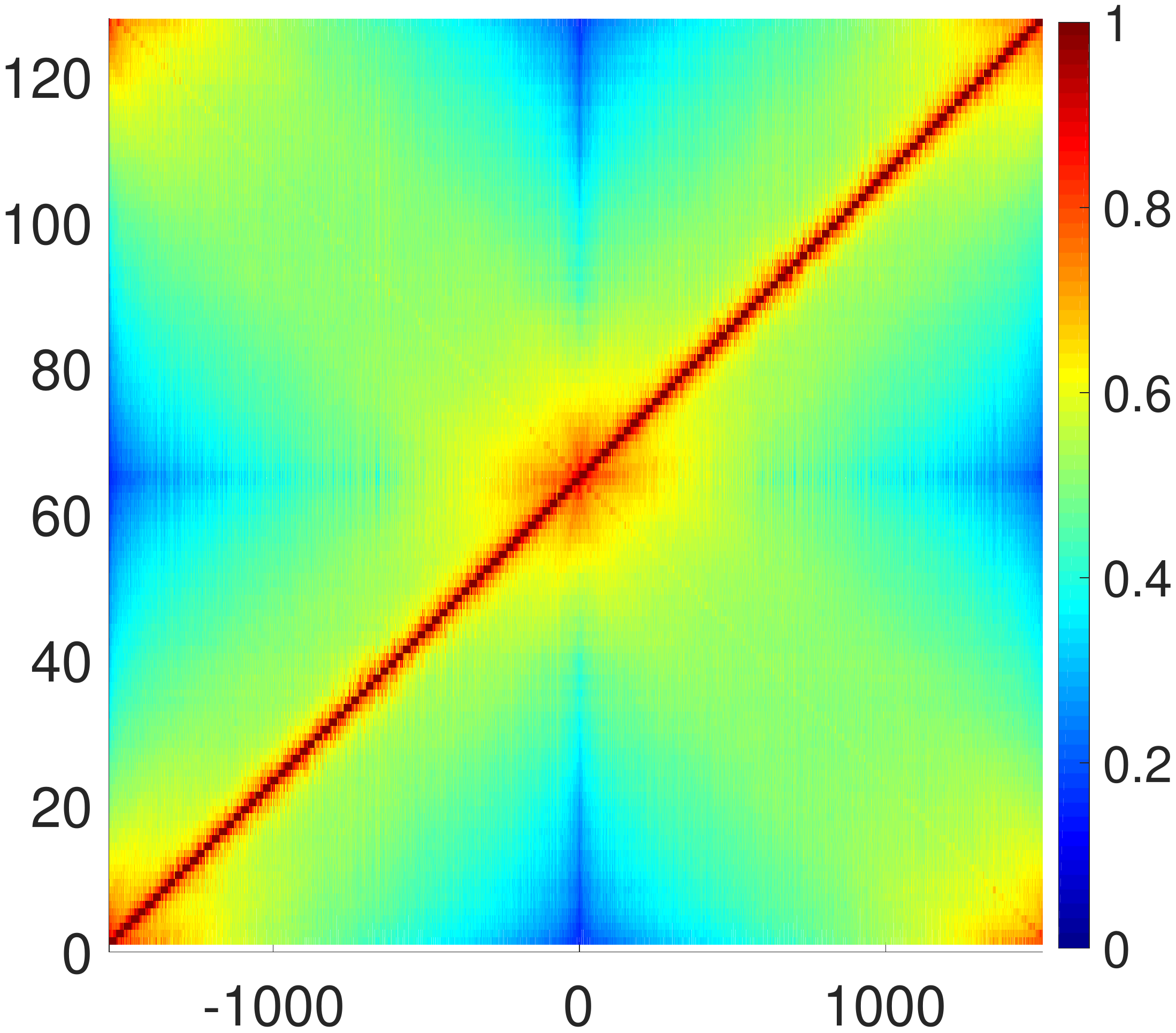}}
    \scalebox{0.25}{\includegraphics{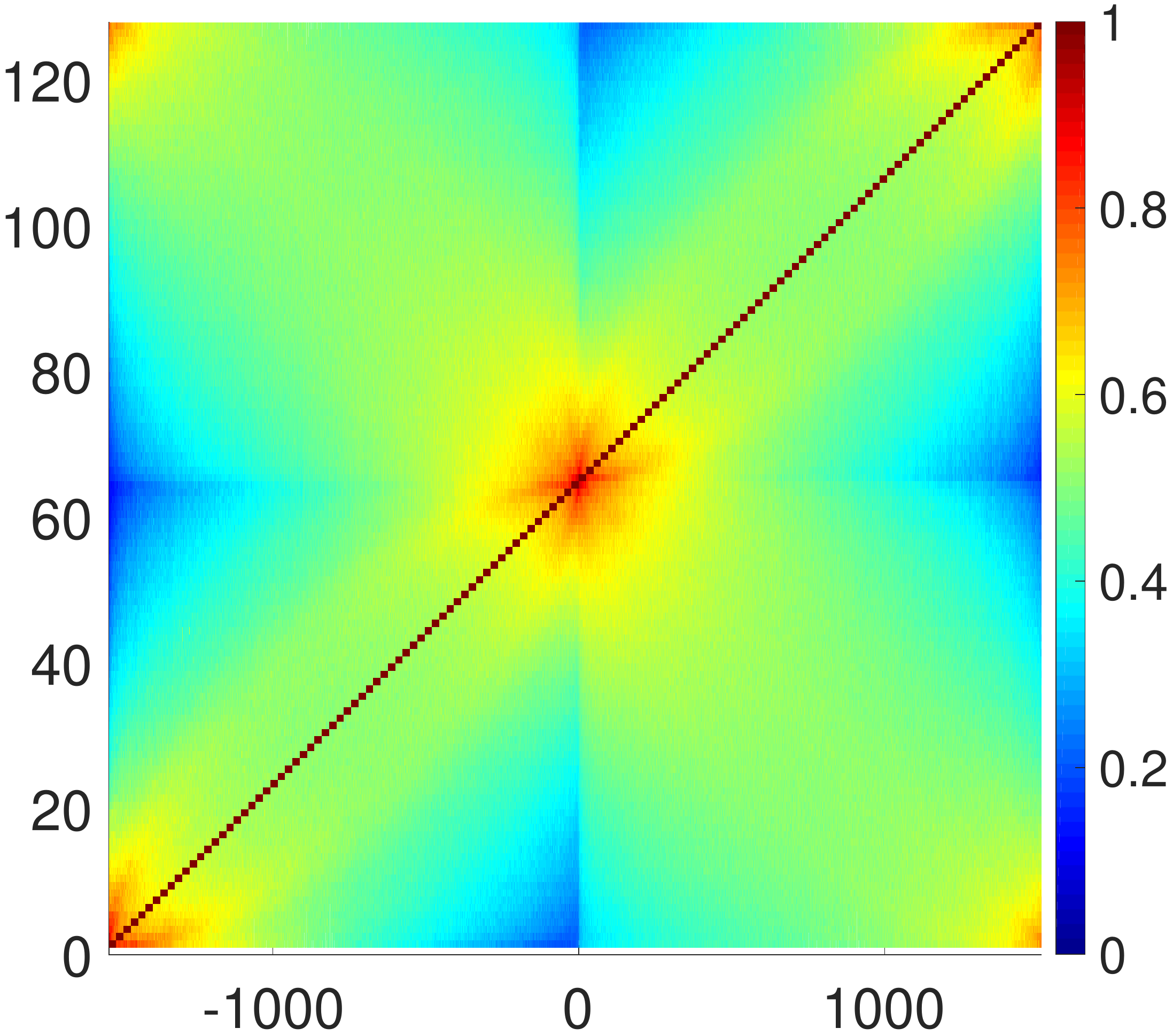}}\\
    \caption{Local coherence for $D$ (left) and $Q$ (right) for a $4^4$ configuration, cf. Table~\ref{tab:confs}.} \label{fig:lc_D_vs_Q}
  \end{center}
\end{figure}

Figure~\ref{fig:lc_D_vs_Q} gives values for $\lc(v)$ for the Wilson-Dirac operator $D$ and the corresponding Hermitian Wilson-Dirac operator $Q$ on a $4^4$ lattice.
Since this lattice is so small, we can compute the full spectrum ($12\cdot 4^4=3072$ eigenpairs) of both $D$ and $Q$. For each matrix we then consider a partitioning of the eigenvectors into 128 sets, each set consisting of 24 consecutive eigenpairs. Here, ``consecutive'' refers to an ordering based on the modulus and the sign of the real part; see the next paragraph for details.
For each of these ``interpolation sets'', the corresponding row displays the color coded value of $\lc(v)$ when projecting an eigenvector $v$ with the projection $\Pi$ corresponding to the aggregation-based interpolation $P$ built with the eigenvectors from that interpolation set as test vectors. The aggregates used were based on a decomposition of the $4^4$ lattice into $16$ sub-lattices of size $2^4$. Due to the spin structure preserving approach, we have two aggregates per sub-lattice, each built from the corresponding spin components of the $24$ test vectors\footnote{The projection $\Pi$ therefore projects onto a subspace of dimension $24 \cdot 2 \cdot 16 = 768$. If there were no local coherence at all, the expected value of $\lc$ is thus $768/(12 \cdot 4^4) = 0.25$.}. Of course $\lc(v)=1$ (dark red) if $v$ is from the respective interpolation set.

The numbering of the eigenvalues used in these plots is as follows:
The plot for $D$ has its eigenvalues with negative imaginary part in its left half, ordered by descending modulus and enumerated by increasing, negative indices including zero, $-1\,535, \ldots, 0$. The eigenvalues with positive imaginary part are located in the right half, ordered by ascending modulus and enumerated with increasing positive indices $1,\ldots,1536$.  
For $Q$ we just order the real eigenvalues by the natural ordering on the reals, using again negative and positive indices. Thus, for $D$ as for $Q$, eigenvalues small in modulus are in the center and their indices are small in modulus, while eigenvalues with large modulus appear at the left and right ends and their indices are large in modulus.

Although one must be careful when drawing conclusions from extremely small configurations, Figure~\ref{fig:lc_D_vs_Q}  illustrates two important phenomena. Firstly, local coherence appears for both $D$ and $Q$, but it is more pronounced for the non-Hermitian Wilson-Dirac matrix. This especially holds directly next to the interpolation sets (the diagonal in the plots).
Secondly, local coherence is particularly strong and far-reaching when projecting on the interpolation sets corresponding to the smallest and largest eigenpairs in absolute values. In the center of both plots, we observe a star-shaped area with particularly high local coherence. This area corresponds to around $10\%$ of the smallest eigenvalues. 
To a lesser extent, local coherence is also noticeable for the other parts of the spectrum, as we consistently observe higher values for $\lc(v)$ for eigenvectors close to the respective interpolation set.

The right part of Figure~\ref{fig:lc_D_vs_Q_large} reports similar information for the Hermitian Wilson-Dirac operator $Q$ coming from a larger, realistic configuration on a $64 \times 32^3$ lattice.  For lattices of this size we cannot compute the full spectrum, thus we show the values for the $984$ smallest eigenpairs, subdivided in 41 interpolation sets, each consisting of 24 consecutive eigenpairs. The aggregates were this time obtained from $4^4$ sublattices. For comparison, the left part of the figure shows a zoomed-in part of the local coherence plot for $Q$ for the $4^4$-lattice from Figure~\ref{fig:lc_D_vs_Q}. 

\begin{figure}[ht]
  \begin{center}
    \scalebox{0.23}{\includegraphics{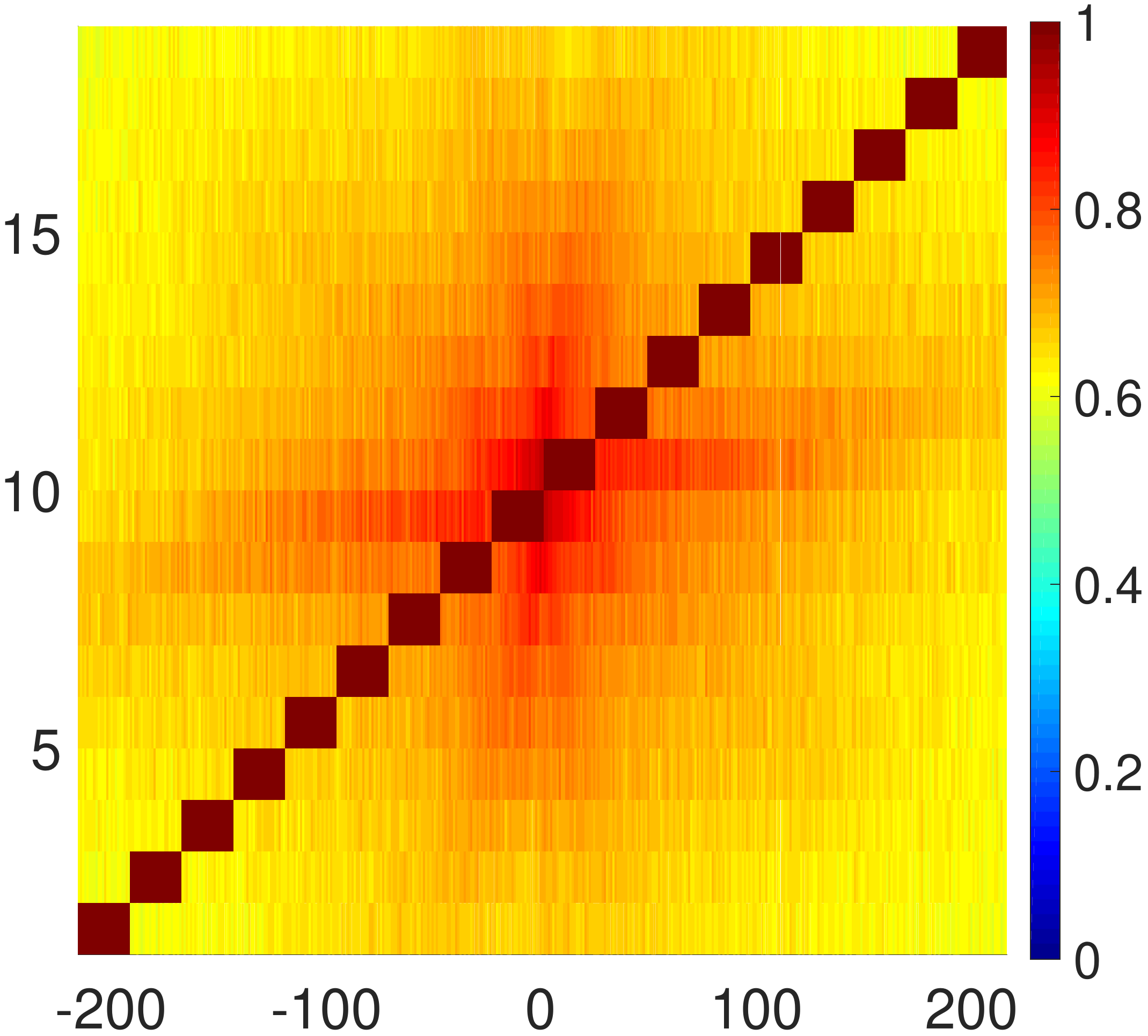}}
    \scalebox{0.23}{\includegraphics{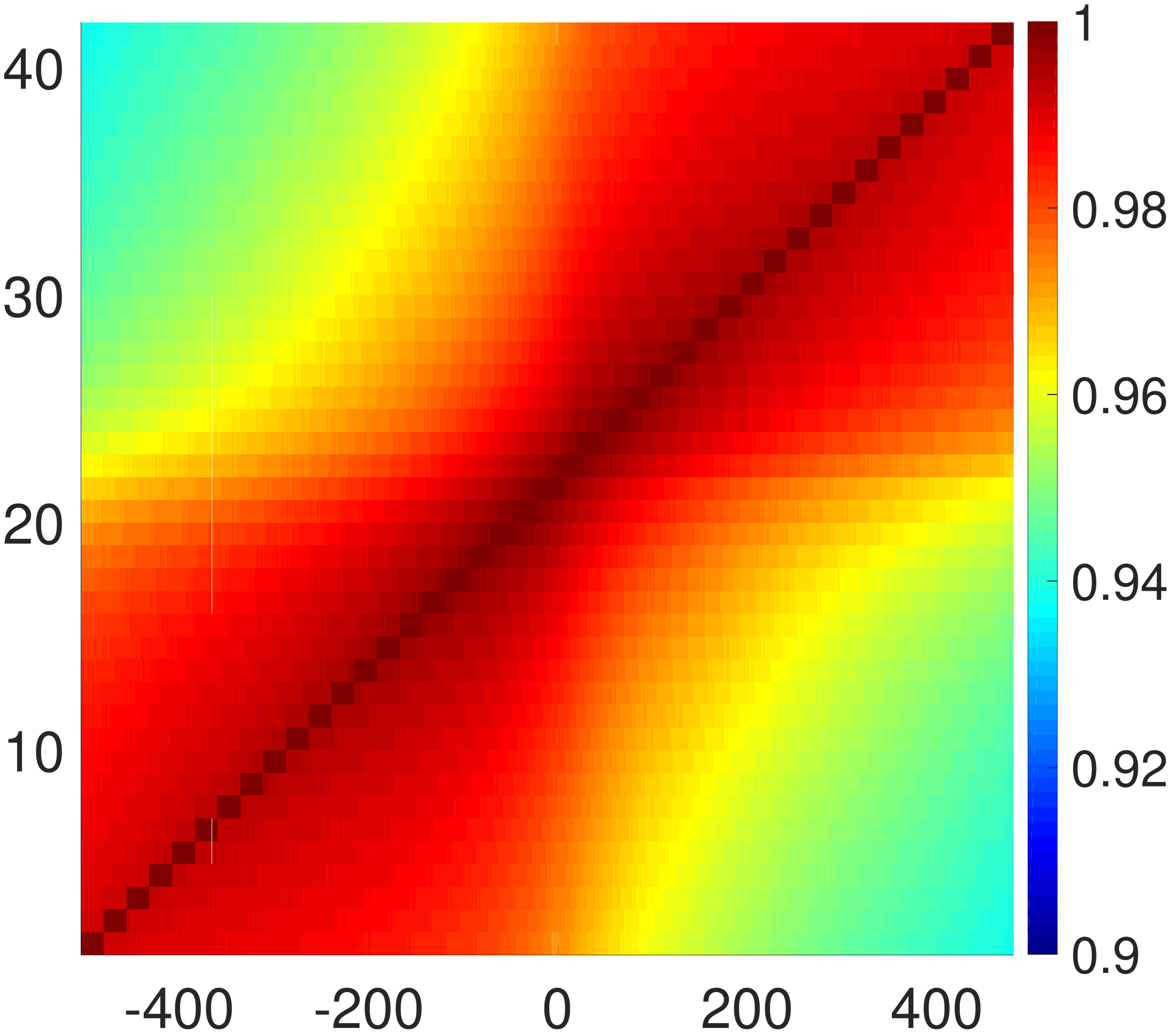}}\\
    \caption{Local coherence for $Q$ for different lattices, focusing on the eigenpairs closest to zero. \emph{Left:} $432$ eigenpairs of a $4^4$ lattice. \emph{Right:} $984$ eigenpairs of a $64\times32^3$ lattice.} \label{fig:lc_D_vs_Q_large}
  \end{center}
\end{figure}

First note that the colors encode different values in the left and right part of Figure~\ref{fig:lc_D_vs_Q_large}.
Local coherence does not drop below $0.9$ for the large configuration, while for the small configuration it goes down to $0.6$. On the other hand, $984$ eigenvalues only correspond to a minuscule fraction of roughly $4\cdot 10^{-3}\%$ of the total of  $12\cdot 32^3\cdot64 = 25\,165\,824$ eigenvalues, which is much less than the roughly $10\%$ depicted for the small configuration. For the interpolation operator of the large configuration we used aggregates corresponding to $4^4$ sub-lattices, which give a total of $2\times 2^{13} = 16\, 384$ aggregates. In relative terms, this is several orders of magnitude finer as for the $4^4$ lattice. This finer aggregation leads to interpolation operators with increased faculties to recombine information, which explains the resulting higher local coherence.

Both parts of Figure~\ref{fig:lc_D_vs_Q_large} show that local coherence drops off for eigenvectors farther away from the interpolation set. For the large configuration, we see, for example, that local coherence of the vectors from the last interpolation set with the second-to-last interpolation set is very high as indicated by the deep red color in the top right corner of the plot. The local coherence of these vectors with respect to the central interpolation set (which contains the eigenpairs with eigenvalues closest to 0) is significantly smaller, indicated by the yellow color at the middle of the right-hand boundary of the plot.    
In the scenario where we choose the shift $\tau$ in the correction equation~\eqref{eq:g5_prec} farther away from zero---as we are targeting eigenpairs close to $\tau$---a coarse grid operator constructed using the eigenpairs closest to zero thus becomes increasingly less effective, reducing the overall convergence speed of the multigrid method significantly. 
To remedy this, we propose a dynamical interpolation updating approach, resulting in a coarse grid operator that remains effective on the span of the eigenvectors with eigenvalues close to the value of $\tau$ set in the outer iteration of the generalized Davidson method. 
In the course of the outer iteration, once enough eigenpairs are available, we therefore rebuild the interpolation, and with it the coarse grid operator, using the already converged eigenvectors which are \emph{closest} to the currently targeted harmonic Ritz value.
Once a harmonic Ritz value converged to an eigenvalue and we choose a new target value $\tau$ that has the same sign as the previous target, we replace one eigenvector from the interpolation set---the one farthest away from $\tau$---with the newly converged eigenvector, and update the multigrid hierarchy. If the new $\tau$ has its sign opposite to the previous one we replace the full interpolation set with converged eigenvectors closest to the new $\tau$, and again update the multigrid hierarchy. 
The updates of the interpolation and coarse grid operators involve some data movement and local operations, but their cost is minor compared to the cost of the other parts of the computation.

With this approach, the coarse grid is always able to treat the eigenspace closest to our current harmonic Ritz approximation efficiently and makes optimal use of the existing local coherence. 
This results in a significantly faster multigrid solver when larger shifts are used, i.e., when a large number of small eigenpairs has to be computed. Since the solution of these shifted systems accounts for most of the work in the eigensolver this approach improves the eigenvalue scaling to a nearly linear complexity, as seen in Section~\ref{sec:tests}.

\begin{algorithm}[H]  \caption{GD-$\lambda$AMG}\label{alg:gd2}
  \SetAlgorithmStyle
    \Input{Hermitian Dirac operator $Q$, no. of eigenvalues $n$, no. of test vectors $n_{tv}$, min. and max. subspace size $m_{min}$ and $m_{max}$, initial guess $[v_1,\ldots,v_{n_{tv}},t]=:[V|t]$, desired accuracy $\varepsilon_{outer}$}
    \Output{set of $n$ eigenpairs $(\Lambda,X)$}
    $\Lambda = \emptyset$, $X=\emptyset$\;
    \For{$m=n_{tv}+1,n_{tv}+2,\ldots$} {
      $t = (I - VV^H)t$, $t = (I - XX^H)t$\;
      $ v_m = t/||t||_2$\;
      $V = [V|v_m]$ \;
      get all $(\theta_i,s_i)$ with $(QV)^H(QV)s_i = \theta_i(QV)^HVs_i$\;
      find smallest (in modulus) $\theta_i\notin\Lambda$\;
      $u = Vs_i$, $r = Qu-\theta_i u$\;
      \If{$||r||_2 \leq\varepsilon_{outer}$} { 
        \tcp{current eigenpair has converged}
        $\Lambda = [\Lambda,\theta_i]$, $X=[X,u]$ \;
        update smallest (in modulus) $\theta_i\notin\Lambda$\;
        $u = Vs_i$, $r = Qu-\theta_i u$\;
        rebuild interpolation using the $n_{tv}$ eigenvectors $x_j$ with eigenvalue $\lambda_j$ closest to $\theta_i$ and update multigrid hierarchy\;
      }
      \tcp{solve correction equation}
      $t = \texttt{DD-}\alpha\texttt{AMG}((D-\theta_i\Gamma_5),\Gamma_5r)$\;\label{alg:gd2:solve}
      \tcp{restart}
      \If{$m\geq m_{max}$ } {
        get $(\Theta, S)$ as all eigenpairs $(\theta_i,s_i)$ of $(QV)^H(QV)s_i = \theta(QV)^H Vs_i$\;
         sort $(\Theta, S)$ by ascending modulus of $\Theta$\;
        \For{$i=1,\ldots, m_{min}$} {
          $V_i = Vs_i$\;
          $(QV)_i = (QV)s_i$\;
        }
          retain first $m_{min}$ vectors of $V$ and $QV$\;
      }
    }
\end{algorithm}

A summary of our eigensolver, termed GD-$\lambda$AMG, a generalized Davidson method with algebraic multigrid acceleration, is given as Algorithm~\ref{alg:gd2}.

\section{Numerical Tests}\label{sec:tests}
In this section we present a variety of numerical tests to analyze the efficiency of the GD-$\lambda$AMG eigensolver. 

Table~\ref{tab:confs} contains information on the gauge configurations we use in our tests. 
The two small configurations on a $4^4$ and a $8^4$ lattice were generated using our 
own heat-bath algorithm. The configurations on the larger lattices (configurations $1$ to $4$) 
were provided by our partners at the University of Regensburg within the Collaborative Research Centre SFB-TRR55; see \cite{Bali:2014gha}. For these configurations, 
we actually use \emph{clover improved} (see \cite{Sheikholeslami:1985ij}) Wilson-Dirac operators $D$ and $Q$, where a block diagonal term  with $6\times 6$ diagonal blocks is added to improve the lattice discretization error from 
$\mathcal{O}(a)$ to $\mathcal{O}(a^2)$. The resulting modified $D$ is still 
$\Gamma_5$-Hermitian. The mass parameter $m_0$ from \eqref{Wilson-Dirac:eq} is chosen such that $a \cdot m_0=\frac{1}{2\kappa}-4$ for computations with configurations $1$ to $4$. For the small configurations we chose $m_0$ such that we obtain a comparable conditioning of the matrix.
A second table, Table~\ref{tab:params}, shows the default algorithmic parameter settings we used within GD-$\lambda$AMG. For a more detailed explanation of these parameters we refer to \cite{Frommer:2013fsa}.

\begin{table}
  \begin{center}
    \begin{tabular}{cccccr}
      \toprule 
      ID  & lattice size      & hopping parameter & mass parameter & clover term & CPU    \\
          & $N_t\times N_s^3$ & $\kappa$          & $a\cdot m_0$        & $c_{sw}$    & cores  \\
      \midrule
      $1$ & $48\times24^3$    & 0.13620           & $-0.3289$      & $1.9192$ & $648$     \\ 
      $2$ & $64\times32^3$    & 0.13632           & $-0.3322$      & $1.9192$ & $1,\!024$ \\
      $3$ & $64\times40^3$    & 0.13632           & $-0.3322$      & $1.9192$ & $2,\!000$ \\
      $4$ & $64\times64^3$    & 0.13632           & $-0.3322$      & $1.9192$ & $4,\!096$ \\
      $5$ & $ 4\times 4^3$    &  --               & $-0.7867$      & $0$      & --        \\
      $6$ & $ 8\times 8^3$    &  --               & $-0.7972$      & $0$      & --        \\
      \bottomrule
    \end{tabular}
  \end{center}
  \caption{Configurations used for numerical tests. Configurations $1$--$4$ are configurations from the ensembles II and IV--VI from~\cite{Bali:2014gha}. 
  Configurations $5$ and $6$ were generated locally and are mainly used for small scale MATLAB experiments.}
  \label{tab:confs}
\end{table}

\begin{table}
  \begin{center}
    \begin{tabular}{rllc}
      \toprule
                           & parameter                            & symbol                       & default         \\
      \midrule
      DD-$\alpha$AMG setup & number of test vectors               & $n_\mathit{tv}$              & $24$            \\
                           & setup iterations                     &                              & $6$             \\
                           & (post-)smoothing steps               &                              & $4$             \\
      \midrule
      DD-$\alpha$AMG solve & relative residual                    & $\varepsilon_\mathit{inner}$ & $10^{-1}$       \\
                           & maximum iterations                   &                              & $5$           \\
                           & coarse grid tolerance                &                              & $5\cdot10^{-1}$ \\

      \midrule
      eigensolver method   & relative eigenvector residual        & $\varepsilon_\mathit{outer}$ & $10^{-8}$       \\
                           & number of eigenpairs                 &                              & $100$           \\
                          & minimum subspace size                 & $m_{min}$                              & $30$           \\
                           & maximum subspace size                & $m_{max}$                             & $50$           \\
      \bottomrule

    \end{tabular}
  \end{center}
  \caption{List of algorithmic \emph{default} parameters.}
  \label{tab:params}
\end{table}

The numerical results involving configurations $1$--$4$ were obtained on the JURECA and JUWELS clusters at the J\"ulich Supercomputing Centre~\cite{wwwJURECA, wwwJUWELS}, while results involving the other lattices were obtained on a smaller workstation.
We will compare our results with the state-of-the-art library PRIMME and with PARPACK, and we start by outlining their underlying basic algorithms.\\

\textbf{PRIMME} 
(PReconditioned Iterative MultiMethod Eigensolver) \cite{PRIMME, svds_software} implements a broad framework for different Davidson-type eigensolvers. Its performance is best if it is given an efficient routine to solve linear systems with the matrix $A$, and we do so by providing the $\Gamma_5$-preconditioned DD-$\alpha$AMG solver. %
There are two key differences compared to GD-$\lambda$AMG:
\begin{itemize}
  \item The interpolation cannot be updated efficiently within the PRIMME framework (at least not without expert knowledge on the underlying data structures), hence we do not update it for this method.
  \item PRIMME uses a Rayleigh Ritz instead of a harmonic Ritz approach to extract eigenvalue approximations.
\end{itemize}

PRIMME has a fairly fine-tuned default parameter set, e.g., for subspace size or restart values, and is able to dynamically change the eigensolver method. We keep the default settings and provide the same multigrid solver to PRIMME as we do for GD-$\lambda$AMG.\\

\textbf{PARPACK} (Parallel ARnoldi PACKage) \cite{wwwPARPACK} is a somewhat older but widely used software for the computation of eigenvalues of large sparse matrices. It is based on an implicitly restarted Arnoldi method, which is originally designed to find extremal eigenvalues.
It is possible to transform an interior problem into an exterior one using a filter polynomial, i.e., a polynomial which is large on the $k$ interior eigenvalues we are looking for and small on the remaining ones. To construct such a polynomial, for example as a Chebyshev polynomial, we need information on the eigenvalue $\lambda_{max}$ which is largest in modulus and the $(k+1)$st smallest in modulus, $\lambda_{k+1}$. 
While $\lambda_{max}=8$ is a sufficiently good estimate for the Hermitian Wilson-Dirac matrix $Q$, no a-priori guess for $\lambda_{k+1}$ is available in realistic scenarios.
For our tests, we run one of the other methods to compute the first $k$ eigenvalues and then use a slightly larger value as a guess for $\lambda_{k+1}$.
While this approach obviously costs a lot of additional work and actually makes the subsequent Arnoldi method obsolete, it is a good reference for a \emph{near-optimally} polynomially filtered Arnoldi method.
Since this approach does not require inversions of the matrix $Q$, the parameter set for this method is rather small. We use a degree ten Chebyshev polynomial as the filter polynomial and set the maximum subspace size to be twice the number of sought eigenpairs. The required eigenvector residual is set to $10^{-8}$, as with the other methods.

\subsection{Algorithmic tuning}
\paragraph{Solving the correction equation.}

Each step of Algorithm~\ref{alg:gd2} uses DD-$\alpha$AMG in line~\ref{alg:gd2:solve} to solve the $\Gamma_5$-preconditioned correction equation $(D-\tau \Gamma_5)t = \Gamma_5r$. More precisely, as indicated by the parameters given in the middle of Table~\ref{tab:params}, we stop the outer (FGMRES) iteration of DD-$\alpha$AMG once the initial residual norm is reduced by a factor of $0.1$ or a maximum of 5 iterations is achieved. Within each DD-$\alpha$AMG iteration we require a reduction of the residual by a factor of 0.5  when solving the system on the coarsest level. 
 Table~\ref{tab:g5_prec} shows that the $\Gamma_5$-preconditioning yields indeed significant gains in compute time.

\begin{table}[ht]
  \begin{center}
    \begin{tabular}{lrrr}
      \toprule
      correction equation               & \multicolumn{2}{c}{iterations} & Time       \\
                        & outer          & inner         & in core-h. \\
      \midrule
      Eq.~\eqref{eq:no_g5_prec}: $(Q-\tau I)t=r$               &           565  &        10,349 &     83.0   \\
      Eq.~\eqref{eq:g5_prec}: $(D-\tau\Gamma_5)t=\Gamma_5r$ &           511  &         3,045 &     41.3   \\
      \bottomrule
    \end{tabular}
  \end{center}
\caption{Impact of $\Gamma_5$-preconditioning for the computation of $100$ eigenpairs of configuration $1$ (see Table~\ref{tab:confs}).}\label{tab:g5_prec}
\end{table}

A variant of generalized Davidson methods solves, instead of the correction equation~\eqref{eq:correction}, the Jacobi-Davidson projected ~\cite{jacobidavidson} system $(I-uu^H)(A-\theta I)(I-uu^H)$, where $u$ is the last (harmonic) Ritz vector approximation. This will avoid stagnation in the case that the correction equation is solved \emph{too exactly}. There are theoretically justified approaches which adaptively determine how accurately the projected system should be solved in each iteration.
Since we solve the correction equation to quite low relative precision ($10^{-1}$ only), we could not see a benefit from using the Jacobi-Davidson projected system. Indeed, even with the adaptive stopping criterion, this approach increased the compute time by approximately 15\%.

\paragraph{Impact of the smoother}
The original DD-$\alpha$AMG method uses SAP as a smoother and we have shown in Chapter~\ref{sec:ddamg} that SAP is also applicable for the Hermitian Wilson-Dirac operator $Q$, yielding the same error propagation operator as long as the individual block systems are solved exactly. 
We now compare a cost-efficient, approximate SAP and GMRES as smoothers within the multigrid methods constructed for the matrices $D-\tau I$, $Q-\tau I$ and $D-\tau\Gamma_5$, where $\tau$ ranges from $0$ to $0.5$ for configuration $6$.
Note that $D-\tau I$ is not relevant for this work, since it would arise when computing eigenpairs for $D$. We still include the results here to be able to compare the performance of DD-$\alpha$AMG for $Q-\tau I$ and $D-\tau\Gamma_5$ with the performance of DD-$\alpha$AMG for $D-\tau I$.

\begin{figure}[ht]
  \center\scalebox{.54}{\includegraphics{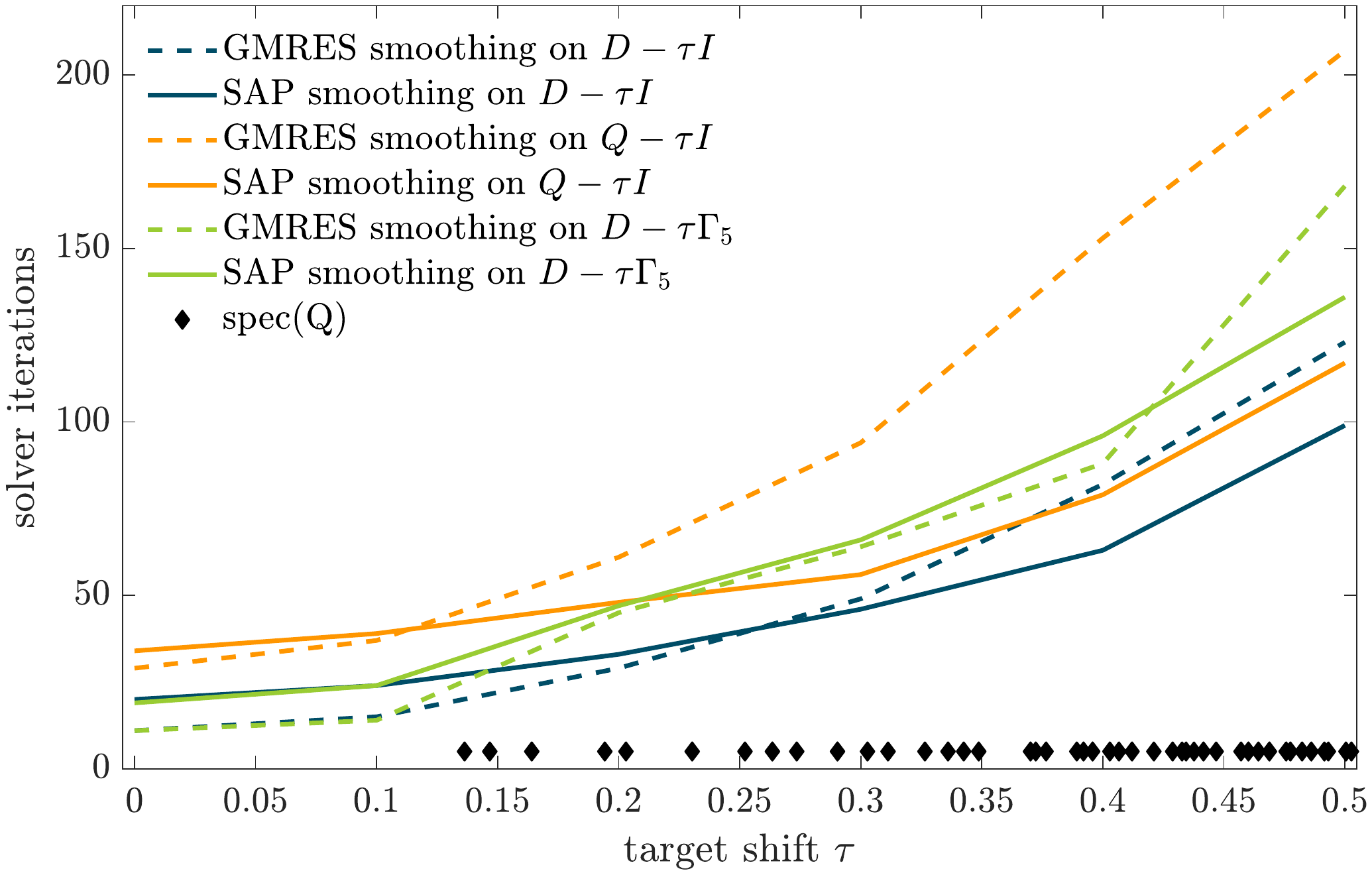}}
  \caption{Comparison of iteration counts of the DD-$\alpha$AMG method using either SAP or GMRES smoothing for configuration $6$ and increasing target shifts $\tau$. The black diamonds at the bottom depict the eigenvalue distribution of $Q$.}
  \label{fig:smoother}
\end{figure}

Figure~\ref{fig:smoother} shows a scaling plot with respect to the target shift $\tau$ for configuration $6$.
For this plot, we used a two-level DD-$\alpha$AMG method with six steps of the adaptive setup procedure to generate the coarse grid system. {The GMRES smoother requires a reduction of the residual by a factor of $10^{-1}$ or after a maximum of ten iterations have been performed. Similarly, the SAP smoother performs three sweeps of SAP, where each block solve is performed using GMRES until a reduction of the residual by a factor of $10^{-1}$ was achieved for the individual block of after a maximum of ten iterations have been performed.
This way the computational work for both smoothers is roughly comparable.}

Figure~\ref{fig:smoother} verifies what was stated in Section~\ref{sec:gdlamg}, namely that DD-$\alpha$AMG converges more slowly for $Q$ compared to $D$.
It also shows that $\Gamma_5$-preconditioning is beneficial in the case of GMRES smoothing whereas in the case of SAP smoothing, it loses efficiency compared to $Q$, although only by a small margin. 
Comparing the two smoothing methods for $D-\tau\Gamma_5$, we see that both methods perform nearly identical up to larger shifts, where SAP starts to be slightly more favorable. We do not expect this to be relevant for larger configurations, though, since there the spectrum is much more dense. Even when aiming for a large number of small eigenvalues, we certainly do not expect to end up with $\tau$-values as large as $0.2$ already.
Since the focus of this paper is on finding an efficient coarse grid operator, and not on optimizing the smoother, we stick to GMRES smoothing here. Implementing SAP instead of GMRES for $D-\tau\Gamma_5$ within the DD-$\alpha$AMG framework would require a more substantial remodeling of the DD-$\alpha$AMG code.

\paragraph{Impact of the coarse grid correction}
For an assessment of the impact of the coarse grid correction step we compute $100$ eigenvalues for configuration $1$, once using DD-$\alpha$AMG with GMRES smoothing to solve the correction equation, and once with a modification where we turned-off the coarse grid correction. This yields a generalized Davidson method where the $\Gamma_5$-preconditioned correction equation~\eqref{eq:g5_prec} is solved using FGMRES with the GMRES-steps of the smoother as a non-stationary preconditioner, i.e., GMRESR, the recursive GMRES method~\cite{GMRESrec}. Note that we do not yet include updating the multigrid hierarchy as the outer iterations proceeds.

\begin{figure}[ht]
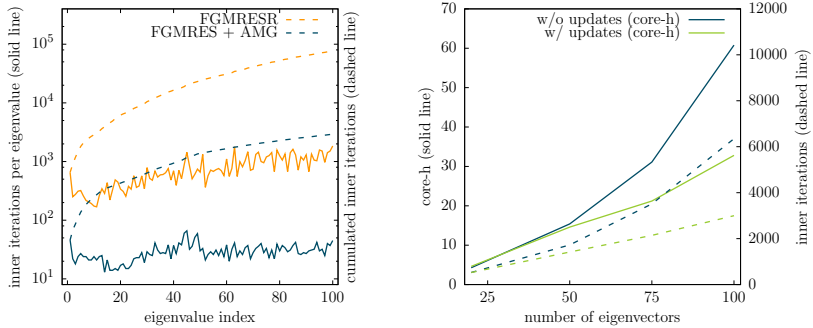

  \center\scalebox{0.5}{\fontsize{13}{15}\selectfont\input{./plots/coarse_grid}\hspace{-3em}\input{./plots/updates_vs_no_updates}}
  \caption{\textit{Left:} Computation of $100$ eigenvalues for configuration $1$ with GMRESR and FGMRES + AMG.
  \textit{Right:} Comparing eigenvalue scaling for configuration $2$ depending on whether eigenvalue information is provided for the interpolation operator.}
  \label{fig:coarse_grid}
\end{figure}

The left part of Figure~\ref{fig:coarse_grid} shows the FGMRES iterations spent on the correction equation for computing $100$ eigenvalues for the two variants.
We see that right from the beginning, including the coarse grid correction, i.e., using the multigrid method, reduces the iteration count by one order of magnitude compared to the ``pure'' GMRESR-Krylov subspace method. 
The required number of FGMRES iterations per eigenvalue stays constant at $\approx 30$ for the multigrid method, whereas GMRESR starts at $\approx 300$ and increases to $\approx 1,200$ for the last eigenvalues.
This is also reflected in CPU time, where on JUWELS multigrid preconditioning results in  $30$ core-h for the entire computation, whereas $217$ core-h were necessary when using GMRESR.
Thus multigrid gains one order of magnitude and, in addition, shows an improved scaling behaviour, despite the loss of local coherence for the larger eigenvalues.

The right part of Figure~\ref{fig:coarse_grid} now illustrates the additional benefits that we get from turning on the updating of the multigrid hierarchy, i.e., when performing full GD-$\lambda$AMG as described in Algorithm~\ref{alg:gd2}. 
Both approaches perform similarly as long as a small amount of eigenvalues is sought.
This changes substantially for already a moderate amount of eigenvalues to a point where interpolation updates save roughly a factor of two in both, number of iterations (dashed lines) and consumed core-h (solid lines).
In terms of iterations it is also noteworthy that interpolation updates lead to a nearly linear scaling with respect to the eigenvalue count, whereas in the other case the scaling is closer to quadratic.

\subsection{Scaling results}
\paragraph{Scaling with the lattice size.}\label{sec:lat_scaling}
We now compare GD-$\lambda$AMG, PRIMME and PARPACK in terms of scaling with respect to the lattice size. 
For this, we report the total core-h consumed for computing $100$ eigenpairs on configurations $1$ to $4$.

\begin{figure}[ht]
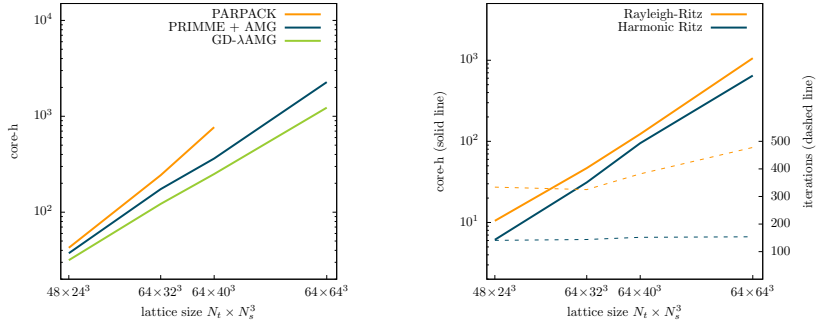

  \center\scalebox{0.5}{\input{./plots/lat_scaling}\hspace{-3em}\input{./plots/rr_vs_hr}}
  \caption{\textit{Left:} Computation of $100$ eigenvalues for $48\times24^3$ to $64\times64^3$ lattices for different methods. \textit{Right:} Comparison of Ritz and harmonic Ritz eigenpair extraction for different lattice sizes.}
  \label{fig:lat_scaling}
\end{figure}

Figure~\ref{fig:lat_scaling} shows that PRIMME and GD-$\lambda$AMG scale similarly with increasing lattice size. GD-$\lambda$AMG shows some improvement in core-h compared to PRIMME, and this improvement tends to get larger when increasing the lattice size.
The right part of Figure~\ref{fig:lat_scaling} shows, that this improvement might be partially attributed to the fact that we use a harmonic Ritz extraction. Here, we compare GD-$\lambda$AMG with its default harmonic Ritz extraction to a variant where we use the standard Rayleigh-Ritz extraction as is done in PRIMME. The Figure shows that harmonic Ritz extractions result in substantially less inner iterations. This also yields savings in computational time, which are smaller, due to the additional cost for the inner products. Note that for larger lattices eigenvalues become more clustered. The harmonic Ritz extraction is then more favorable compared to the Rayleigh-Ritz approach, since it is able to better separate the target eigenvalue from the neighboring ones.
PARPACK scales worse than the other methods, even when we use an unrealistic ``near optimal'' filter polynomial as we did here. In practice, i.e., when no guess for $|\lambda_{k+1}|$ is available, PARPACK's performance would fall even further behind.
Applying PARPACK to $Q^{-1}$ to make use of the efficient multigrid solver is way too costly, due to the necessity of very exact solves to maintain the Krylov structure.

\paragraph{Scaling with the number of eigenvalues.}\label{sec:ev_scaling}
Figure~\ref{fig:ev_scaling} reports results of a scaling study obtained for configuration~2. We just compare GD-$\lambda$AMG and PRIMME, since PARPACK is not competitive. 

\begin{figure}[ht]
  \center\scalebox{0.7}{\fontsize{13}{15}\selectfont\input{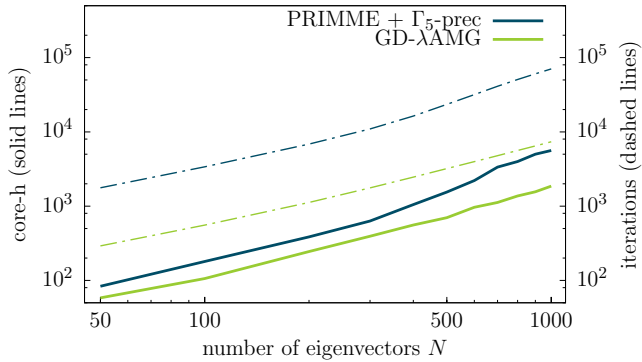}}
  \caption{Eigenvalue scaling in the range of $50$ to $1000$ eigenpairs on configuration $2$ with a lattice size of $64 \times 32^3$.}
  \label{fig:ev_scaling}
\end{figure}

The figure shows that GD-$\lambda$AMG has an advantage over PRIMME when larger numbers of eigenvalues are sought. GD-$\lambda$AMG needs up to one order of magnitude less iterations, which translates to a speed-up of $1.5$ for $50$ eigenvalues to up to more than three for $1\,000$ eigenvalues. This shows that the additional effort due to the adaptive construction of the multigrid hierarchy and the harmonic Ritz extraction is beneficial with respect to the overall performance.
GD-$\lambda$AMG and PRIMME both scale nearly linearly with respect to the number of eigenvalues sought, up to at least $300$ eigenvalues. Then PRIMME's performance starts to decrease more significantly compared to GD-$\lambda$AMG. We see that the increase in the overall computing time in PRIMME scales more than linearly with the number of iterations to be performed. This indicates that the non-adaptive multigrid solver used in PRIMME is getting increasingly less efficient, a situation that is remedied with the update strategy realized in GD-$\lambda$AMG.

\section{Spectral gap} \label{sec:gap}
This paper has a clear focus on algorithmic development to compute small eigenmodes. Fur purposes of illustration, we now include an example which relates our computed eigenvalues to the general approach of lattice QCD.

The \emph{spectral gap} is the relevant quantity for the stability of Monte Carlo simulations of lattice QCD \cite{DelDebbio:2005qa}.
It is defined as the smallest eigenvalue in magnitude of the Hermitian Wilson-Dirac operator $Q$.
If chiral symmetry was preserved, the spectral gap would be bounded from below by the bare current-quark mass $m$. 
Since the Wilson-Dirac operator breaks chiral symmetry it is possible that the gap is smaller than $m$.
In the following we adopt the notation of \cite{DelDebbio:2005qa} and
denote the eigenvalues of $Q^2$ by $\alpha_i$.

\begin{figure}[ht] 
  \center\scalebox{0.4}{\includegraphics{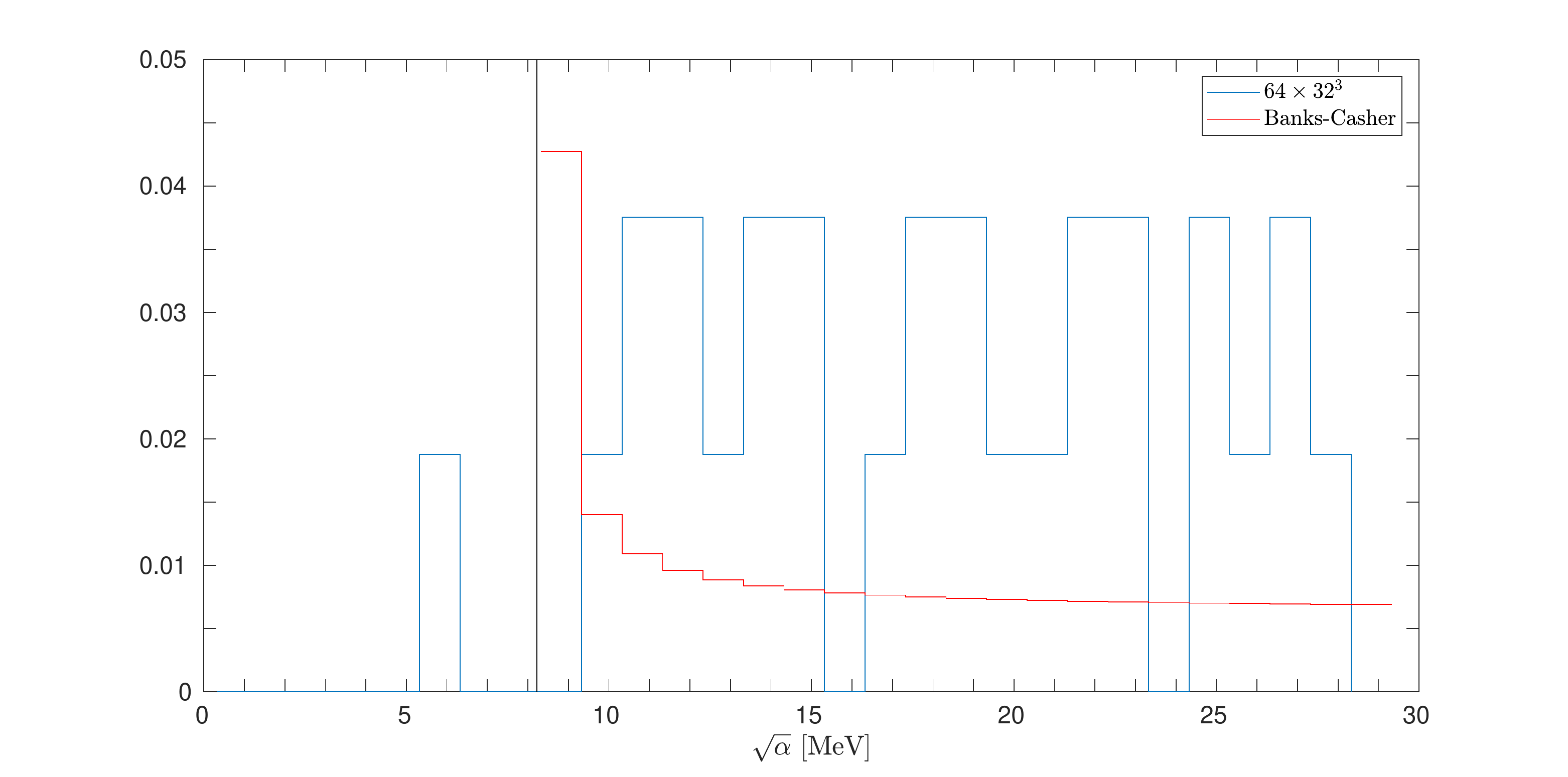}}
  \caption{Number of eigenvalues $\sqrt{\alpha_i}$ of $|Q|$
      per bin and $\mathrm{fm}^4$ at the low end of the spectrum 
      on configuration $2$ with a lattice size of $64 \times 32^3$.
      For comparison the same quantity in the continuum theory
      is also shown, cf. Eq.~\eqref{eq:banks-casher}, taking $m=8.22\,\mathrm{MeV}$ and $\Sigma=(250\,\mathrm{MeV})^3$.}
      \label{fig:ev_gap}
\end{figure}

Figure~\ref{fig:ev_gap} shows the number of eigenvalues $\sqrt{\alpha_i}$ per bin and $\mathrm{fm}^4$ (blue line) at the low end of the spectrum computed on configuration~$2$ with a lattice size of $64 \times 32^3$ and $N_\mathrm{f}=2$ mass-degenerate non-perturbatively O($a$) improved quarks, cf. Table~\ref{tab:confs}.
The lattice spacing is $a=0.071\,$fm from~\cite{Bali:2014gha} and the bin size is set to  $1\,$MeV.

It is interesting to compare these computational results to what we know analytically in the continuum theory in the infinite volume limit.
The spectral density can be computed from the Banks--Casher relation \cite{Banks:1979yr,DelDebbio:2005qa} asymptotically
\begin{equation}\label{eq:banks-casher}
  \tilde{\rho}(\sqrt{\alpha}) \stackrel{\alpha>m^2}{=}
  2\sqrt{\alpha}\left[
    \frac{\Sigma}{\pi\sqrt{\alpha-m^2}} + \mathrm{O}(1) \right] \,,
\end{equation}
which goes to infinity for $\sqrt{\alpha}\rightarrow m$ and is not defined for $\sqrt{\alpha}<m$, resulting in no eigenvalues smaller than $m$ for the continuum case. Note that the Banks-Casher relation can be used to determine the chiral condensate $\Sigma$ from the non-zero density of eigenmodes at the origin in infinite volume; cf.~\cite{Giusti:2008vb}.
In order to compare to our lattice results for configuration $2$ we compute the bare current-quark mass $m$ using \cite[Eq. (E.1)]{Fritzsch:2012wq}. 
The spectral gap (black vertical line in figure~\ref{fig:ev_gap}) is given by $\sqrt{\bar{\alpha}}=Z_{\mathrm{A}}m$~\cite{DelDebbio:2005qa}, which we evaluate using the axial-current renormalization constant $Z_{\mathrm{A}}$ from~\cite{DallaBrida:2018tpn}.
The resulting distribution of eigenvalues $\sqrt{\alpha}$ per bin and $\mathrm{fm}^4$ for the continuum operator (red line) is also plotted in Figure~\ref{fig:ev_gap}.
The comparison to the lattice data shows that on configuration $2$ there is a significant amount of eigenvalues smaller than $\sqrt{\bar{\alpha}}$, and the distribution close to the gap deviates from Equation~\eqref{eq:banks-casher}.
The figure thus illustrates quantitatively the deviations due to lattice artifacts which are not unexpected for the given lattice sizes. 

\section{Conclusion}\label{sec:conclusion}
In this paper we introduced an eigensolver built around the multigrid method DD-$\alpha$AMG for efficient shifted inversions within a generalized Davidson method. Several adaptations were included in order to improve eigenvalue scaling by at least a factor of three for a moderate to large amount of eigenvalues compared to current general purpose eigensolver software. This was accomplished by implementing a synergy between the generalized Davidson method and the AMG solver. Additionally, we incorporated several state-of-the-art techniques like locking, thick restarting and harmonic Ritz extraction to achieve a nearly linear eigenvalue scaling. We included a variety of numerical tests to verify the efficiency of our proposed adaptations.

We can now use the new algorithm to compute relatively many small eigenmodes of large configurations. This allows to advance deflation approaches for stochastic trace estimation as required in the computation of disconnected fermion loops. We plan to explore this further in cooperation with the European twisted mass collaboration. Note that the eigenpairs of the symmetrized Wilson and twisted mass operators differ by an imaginary shift in the eigenvalues only.

With the algorithm presented in this work it is now also possible to perform a systematic study of the spectral gap and the spectral density varying the lattice volume, the lattice spacing, and the quark mass. 
We plan to extend the work of~\cite{DelDebbio:2005qa} to the $N_\mathrm{f}=2+1$
O($a$) improved theory, cf.~\cite{Bruno:2014jqa} in the future.

\subsection{Acknowledgment}\label{sec:thx}
The authors are grateful for the fruitful collaboration with Benjamin M\"uller of the University of Mainz and our colleagues within the SFB-TR 55 from the Theoretical Physics department of the University of Regensburg, especially Gunnar Bali and his group for providing configurations.
We thank the J\"ulich Supercomputing Centre (JSC) for access to high performance computing resources through NIC grant HWU29.

\bibliographystyle{plain}
\bibliography{gdlamg}

\end{document}